\documentclass[12pt]{article}
\usepackage{latexsym}
\usepackage{color}
\usepackage{amssymb}
\usepackage{graphicx}
\usepackage{epsfig}
\usepackage{amsmath}
\usepackage{mathrsfs}
\usepackage{natbib}
\usepackage{hyperref}

\usepackage[utf8]{inputenc}

\renewcommand{\cite}[1]{\citet{#1}}

\newtheorem{theorem}{Theorem}[section]

\newcommand{\cE}{{\cal E}}

\newtheorem{thm}{Theorem}
\newtheorem{pro}{Proposition}
\newtheorem{cor}{Corollary}
\newtheorem{lem}{Lemma}

\newtheorem{defi}{Definition}

\newtheorem{ex}{Example}
\newtheorem{as}{Assumption}

\newcommand{\EE}{\mathbb{E}}

\newcommand{\PP}{{\mathbb{P}}}

\renewcommand{\P}{\PP}

\newcommand{\argmax}[1]{ {\arg\max}_{#1} }

\newcommand{\qed}{\mbox{ }~\hfill~$\Box$ \vspace{1ex} }

\newenvironment{proof}{\noindent{\sc Proof: }}{ \qed }

\newenvironment{proof2}[1]{\noindent{\sc Proof of Theorem #1: }}{ \qed }
\newenvironment{proof3}[1]{\noindent{\sc Proof of Lemma #1: }}{ \qed }
\newenvironment{proof4}[1]{\noindent{\sc Proof of Proposition #1: }}{ \qed }

\newcommand{\rmII}{\text{\it I\kern-.08em I\,}}
\newcommand{\rmIII}{\text{\it I\kern-.08em I\kern-.08em I\,}}
\newcommand{\rmIV}{\text{\it I\kern-.08em V\,}}

\begin{document}

\title{Knight--Walras Equilibria}
\author{  \!Patrick Bei\ss ner\thanks{  Financial Support  through the Research group ``Robust Finance" at the Center for Interdisciplinary Research (ZiF) at Bielefeld university  is gratefully acknowledged.}
\\
\small  Research School of Economics\\[-2pt] \small
The Australian National University \\
\small
patrick.beissner@anu.edu.au
  \and  \:\:\:Frank Riedel\thanks{Frank Riedel is also affiliated with the Faculty   of Economic and Financial Sciences, University of Johannesburg, South Africa.  Support  through DFG Grant Ri 1128--7--1 is gratefully acknowledged.}
\\
\small   Center for Mathematical Economics\\[-2pt] \small
Bielefeld University \\
\small
frank.riedel@uni-bielefeld.de }
\date{First Version: May 14, 2016}
\maketitle

\begin{abstract}
Knightian uncertainty leads naturally to nonlinear expectations. We introduce a corresponding  equilibrium concept  with sublinear prices and  establish  their existence. In general, such equilibria lead to Pareto inefficiency and coincide with Arrow--Debreu equilibria only   if the values of net trades are ambiguity--free in the mean. Without aggregate uncertainty, inefficiencies arise generically.
 We introduce a constrained efficiency concept,  \emph{uncertainty--neutral efficiency} and show that Knight--Walras equilibrium allocations are efficient in this constrained sense. Arrow--Debreu equilibria turn out to be non--robust with respect to the introduction of Knightian uncertainty. 
\end{abstract}

\medskip
{\footnotesize{ \it Key words and phrases:} Knightian Uncertainty, Ambiguity, General Equilibrium \\
{\it \hspace*{0.6cm} JEL subject classification: D81, C61, G11} }

\newpage

\section{Introduction}

Knightian (or model) uncertainty describes the situation in which  the probability distribution of relevant outcomes is not known exactly.  We consider markets where Knightian uncertainty is described by a set of probability distributions. In such a situation, it is natural to work with a nonadditive notion of  expectation derived from the set of probability distributions. We introduce  here a corresponding equilibrium concept, \emph{Knight--Walras equilibrium},  where the forward price  of a contingent consumption plan is given by the maximal expected value of the net consumption value.

In a first step, we establish existence of Knight--Walras equilibrium for general preferences including the well studied classes of smooth ambiguity preferences and  variational preferences.
The proof extends Debreu's game--theoretic existence proof  in an interesting way. Debreu works with a Walrasian player who maximizes the expected value of aggregate excess demand. In our proof, we introduce a further \emph{Knightian} player who chooses the worst probability distribution. Under Knightian uncertainty, we can thus view the ''invisible hand'' as consisting of two auctioneers where one of them chooses the (state) price and the other one the ''relevant'' probability distribution.

In case of pure risk, i.e. when the set of probability distributions consists of a singleton, the new notion coincides with the classic notion of an Arrow--Debreu equilibrium under risk. A main objective of our paper is to study the differences to Arrow--Debreu equilibrium which are created by Knightian uncertainty in prices. 

In a first step, we ask under what conditions Arrow--Debreu and Knight--Walras equilibria coincide. Generalizing the case of pure risk, we show that this holds true if and only if the value of net demands are \emph{ambiguity--free in mean}; in this case, the expected value of all net demands is the same under all probability distributions. 

We then ask how restrictive this condition is. To this end, we study the well--known class of economies without aggregate uncertainty and pessimistic agents. It is well known that agents obtain (efficient) full insurance allocations in this case (\cite{billot2000sharing}). We show that generically in endowments, these Arrow--Debreu equilibria are not Knight--Walras equilibria. Intuitively, it will be rarely the case that agent's net demand is ambiguity--free in mean when individual endowments are subject to Knightian uncertainty. 

One can thus not expect efficiency under Knightian uncertainty. We then study a restricted notion of efficiency which we call \emph{Uncertainty--Neutral Efficiency} which is related to the space of ambiguity--free contingent plans. We show that Knight--Walras equilibrium allocations are uncertainty--neutral efficient.

We then continue to explore the nature of Knight--Walras equilibria in economies without aggregate uncertainty. It turns out that Arrow--Debreu equilibria are not robust with respect to the introduction of Knightian uncertainty in prices. Even with a small amount of Knightian uncertainty, the unique Knight--Walras equilibrium has no trade, in sharp contrast to the full insurance allocation of the Arrow--Debreu equilibrium.

 A discussion on the present type of nonlinear  forward   prices  can be found in \cite{araujo2012pricing} and \cite{beissner2012coherent} where no related questions of equilibrium are addressed.
 Sublinear prices in general equilibrium are studied in \cite{aty00}.
 More recently,  \cite{richter2015back} consider convex equilibria.

The rest of the paper is  organized as follows. Section 2 introduces the concept of Knight--Walras equilibria. Section 3 establishes existence of  equilibria. 	 Section 4  analyzes the relation  to Arrow--Debreu equilibria. Section 5 introduces  and discusses  the alternative Pareto criterion.  Section 6 investigates the equilibrium correspondence when the amount of Knightian uncertainty is a variable.  Section 7 concludes. The Appendix collects  proofs.

\section{Knight--Walras Equilibrium}

\subsection{Expectations and Forward Prices}
We consider a static economy under uncertainty with a finite state space $\Omega$.
In risky environments, or for probabilistically sophisticated agents, expectations are given by probability measures on $\Omega$; under Knightian uncertainty, one is naturally led to sublinear expectations. Let us fix our notion of expectation first.

  Let $\mathbb X=  {\mathbb R}^\Omega$ be the commodity space of contingent plans for our economy.  We call  $\mathbb E : \mathbb X \to \mathbb R$
  a \textit{(Knightian) expectation} if it satisfies the following properties:
  \begin{enumerate}
    \item $\mathbb E$ preserves constants:\quad\:$\mathbb E c = c $ \quad\qquad\qquad\quad \!for all $c \in \mathbb R$,\\[-1.5em]
    \item $\mathbb E$  is  monotone: \quad \qquad \:\: \:$\mathbb  E X \le \EE Y$\quad\quad \quad \qquad \!\!for all  $X, Y \in \mathbb X$, $ X \le Y$,\\[-1.5em]
        \item $\mathbb E$  is sub-additive:\quad \qquad\:$\EE [X+Y] \le \EE X + \EE Y \quad \textnormal{ for all } X, Y \in \mathbb X$,\\[-1.5em]
        \item  $\mathbb E$ is homogeneous: \quad\:\:$
\quad\!\EE\lambda X = \lambda \EE X \qquad \quad\:\: \textnormal{for }\lambda>0  \textnormal{ and }X \in \mathbb X$,\\[-1.5em]
\item $\mathbb E$ is relevant
: \qquad \qquad  \:\:\:$ \mathbb E[ -X] < 0$ \quad\quad \qquad for all $X\in \mathbb{X}_+\setminus \{0\}. $ \\[-1.5em]

  \end{enumerate}

It is well known\footnote{See Lemma 3.5 in \cite{gisc89}, \cite{peng2004nonlinear}, \cite{ar99}, or \cite{follmer2011stochastic}; alternatively, in Robust Statistics, see \cite{huber2011robust}. } that $\mathbb E$ is uniquely characterized by a convex and compact set $\mathbb P$ of probability measures on $\Omega$ such that
  \begin{eqnarray} \label{SL}
\mathbb E X = \max_{P\in\mathbb P} E^P X
 \end{eqnarray}
 for all $X \in\mathbb X$; $E^P$ denotes the usual linear expectation here. Due to the relevance of $\mathbb{E}$, the representing set  
 $\mathbb P$ in \eqref{SL} has full support in the sense that for every $P \in\mathbb P$  we have $P(\omega)>0$ for every $\omega\in\Omega$.




The sublinear expectation $\mathbb E$ leads naturally to a concept of (forward) price for contingent plans: let $\psi: \Omega \to \mathbb R_+$ be a positive state--price. The forward price for a contingent plan  $X \in\mathbb X$ is 
$$\Psi(X)=\mathbb E \psi X\,,$$ in analogy to the usual forward (or risk--adjusted or equivalent martingale measure) price under risk. We call $\Psi:\mathbb{X}\rightarrow\mathbb{R}$ a coherent price system.

\subsection{The Economy with Sublinear Forward Prices}

We now introduce an otherwise classic economy with sublinear forward prices.

\begin{defi}
An \textnormal{Knightian economy} (on $\Omega$) is a triple
$$\cE=\left( I, \left(e^i,U^i \right)_{i\in I}, \mathbb E \right)$$
where
 $I\ge 1$ denotes the number of agents,
  $e^i \in \mathbb X_{+}=\left\{c \in \mathbb X : c(\omega) \ge 0 \, \textit{ for all } \omega \in \Omega \right\}$ is the endowment of agent $i$,    $U^i : \mathbb X_+ \to \mathbb R$ agent $i$'s utility function, and $\mathbb E$ is a Knightian expectation.
\end{defi}

As we fix the agents $\mathbb{I}=\{1,\ldots, I\}$ throughout the paper, we will sometimes use the shorthand notation $\cE^{\EE}$ or $\cE^{\mathbb P}$ to emphasize the dependence of the economy on the Knightian expectation $\EE$ or its set of priors $\mathbb P$.

The following example list some natural   utility functions in $\mathcal{E}$ which have been axiomatized in the literature.
\begin{ex}\label{ExUtility}
\begin{enumerate}
\item Multiple prior expected utilities (\cite{gisc89}) take  the form
\begin{equation}\label{DefMultPriorUt}
U^i(c)=-\mathbb{E}[-u^i(c)] = \min_{P\in\mathbb P} E^P u^i(c)
\end{equation}
 for a suitable (continuous, strictly increasing, strictly concave) Bernoulli utility function $u^i : \mathbb R_+ \to \mathbb R$.

Interpreting  $\mathbb P$ as the objectively given information about Knightian uncertainty,  we can also allow for subjective reactions to such Knightian uncertainty in the spirit of \cite{gajdos2008attitude}
$$U^i(c)= \min_{P\in\phi^i(\mathbb P)} E^P u^i(c)$$ for a selection $\phi^i(\mathbb P) \subset \mathbb P$. Note that a singleton $\phi^i(\mathbb P)=\{P_0\}$ leads to ambiguity--neutral, or subjective expected utility agents. Our model thus includes heterogeneous expectations among agents as a special case.
\item The smooth model of \cite{klibanoff2005smooth}
 has
$$U^i(c)= \int_{\mathbb P} \phi^i\left( E^P u^i(c) \right) \mu(\textnormal{d}P)$$ for a continuous, monotone, strictly concave ambiguity index $\phi^i : \mathbb R \to \mathbb R$ and a second--order prior $\mu$, a measure on $\mathbb P$.
\item  \cite{dana2013intertemporal} introduce  anchored preferences of the form
$$
U^i(c)= \min_{P\in\mathbb P} E^P [u^i(c)-u(e^i)]
$$
to study economies with incomplete preferences. These preferences have recently been axiomatized by \cite{faro2015variational}. Note that we do not treat incomplete preferences here.
\end{enumerate}

\end{ex}

The following equilibrium concept deviates from an Arrow--Debreu equilibrium by the sublinearity in the pricing expectation.
\begin{defi}
 We call a pair $(\psi, c)$ of a state--price $\psi : \Omega \to \mathbb R_+$ and an allocation $c=(c^i)_{i=1,\ldots,I} \in \mathbb X_+^I$ a \emph{Knight--Walras equilibrium} if
\begin{enumerate}
\item  the allocation $c$ is feasible, i.e. $\sum_{i=1}^I (c^i-e^i) \le 0$ 
\item for each agent $i$, $c^i$ is optimal in the  \textnormal{Knight-Walras} budget set
  \begin{eqnarray} \label{budget}
\mathbb B(\psi,e^i)= \left\{c \in \mathbb X_+ : \mathbb E \psi (c-e^i) \le 0 \right\}\,,
  \end{eqnarray}
i.e. if $U^i(d)>U^i(c^i)$  then  $d \notin \mathbb B(\psi,e^i)$. 
\end{enumerate}
\end{defi}

We discuss some immediate properties of the new concept. 
\begin{ex}\label{ex2}
\begin{enumerate}
\item
When $\mathbb P=\{P_0\}$ is a singleton, the budget constraint is linear; in this case, Knight--Walras and Arrow--Debreu equilibria  coincide. In particular, equilibrium allocations are efficient. 
\item At the other extreme, when $\mathbb P = \Delta$ consists of all probability measures, and the state-price $\psi$ is strictly positive, the budget sets consist of all plans $c$ with $c \le e^i$ in all states. We are economically in the situation where all spot markets at time $1$ operate separately and there is no possibility to transfer wealth over states. As a consequence, with strictly monotone utility functions, no trade is an equilibrium for every strictly positive state price $\psi$. Equilibrium allocations are inefficient, in general,  and equilibrium prices are  indeterminate.

\end{enumerate}
\end{ex}

\section{Existence of Knight--Walras Equilibria}

We first establish existence of a Knight--Walras equilibrium. If agents have single--valued demand, one can modify a standard  proof, as, e.g. in \cite{hildenbrand1988equilibrium}, to establish existence.

Under Knightian uncertainty, natural examples arise where demand can be set--valued. A point in case are agents (traders) who minimize a coherent risk measure; equivalently, one might think of ambiguity--averse, yet risk neutral agents. 

 If we include this general case, one needs to work more. We think that the proof, beyond the natural interest in generality, provides additional insights into the working of markets under Knightian uncertainty. We thus provide the more general, lengthier version.

\begin{as}\label{A2}
Each agent's endowment $e^i$  is strictly positive and each utility function $U^i:\mathbb{X}_+\rightarrow \mathbb{R}$ 
\begin{itemize}
\item locally non satiated.
\item  is norm continuous  and concave on $\mathbb{X}_+$.
\\[-1.5em]
\item is monotone:  if $x\geq y$  then $U^i(x)\geq U^i(y)$.
\\[-1.5em]
\end{itemize}
\end{as}

\begin{thm}\label{MAIN}
Under Assumption \ref{A2},  Knight-Walras equilibria  $(\psi,c)$ exist.
\end{thm}

 A standard existence proof of Arrow--Debreu equilibrium  uses a game--theoretic ansatz. One introduces a price player who maximizes the expected value of aggregate excess demand over state prices.  Let us call this type of player a \textit{Walrasian price player}. The consumers  maximize their utility given the budget constraint. The equilibrium of the game is an  Arrow--Debreu equilibrium.

Our method to prove existence follows this game--theoretic approach. Due to Knightian uncertainty, we have to introduce a second,    \textit{Knightian},  price player. This player maximizes the expected value of aggregate excess demand over the priors $P \in \PP$, taking the   state price as given.  The Walrasian price player in the Knight--Walras equilibrium acts in the same way as in the Arrow--Debreu equilibrium.

With nonlinear prices, the question of arbitrage comes up naturally. After all, the equilibrium concept would not be very plausible if it would allow for arbitrage opportunities.  
 
\begin{cor}\label{cor1}
In a Knight--Walras equilibrium, the following no arbitrage condition holds
$$\Psi\left(\sum_i \left(\hat c^i -e^i\right)\right)= \sum_i\Psi\left( \hat c^i -e^i \right).$$
\end{cor}

\begin{proof} In step 6 of the proof of Theorem \ref{MAIN}, we have shown $ \sum_i(\hat c^i-e^i)\leq 0$, local non--satiation implies $ \sum_i(\hat c^i-e^i)= 0$. The constant preserving property of $\mathbb{E}$ yields
\begin{eqnarray*}
0= \mathbb{E}[\psi 0]&=& \mathbb{E}\left[ \psi \sum_i(\hat c^i-e^i)\right]\\
&= & {E}^{P^*}\left[ \psi \sum_i(\hat c^i-e^i)\right] \\
&= & \sum_i {E}^{P^*}\left[ \psi (\hat c^i-e^i)\right] \\
&\leq & \sum_i \mathbb{E}\left[\psi( \hat c^i-e^i)\right]\\
&=&0.
\end{eqnarray*}
The last equality follows  from the fact that each agent exhausts the budget in equilibrium.
\end{proof}

\section{(Non--) Equivalence to Arrow--Debreu Equilibrium}

If the expectation $\mathbb E$ is linear, Knight--Walras equilibria are Arrow--Debreu equilibria; so agents achieve an efficient allocation, by the first welfare theorem.
With incomplete Knightian preferences in the sense of \cite{Bew02}, the Arrow--Debreu equilibria of the linear economies $\cE^{\{P\}}$ are also equilibria under Knightian uncertainty (see \cite{rigotti2005uncertainty} and \cite{dana2013intertemporal}, wo da angeben!). It seems thus natural to ask whether such a result might hold true for our Knightian economies.

We will now explore under what conditions this equivalence remains. In a first step, we show that Knight--Walras equilibria are Arrow--Debreu equilibria if and only if the net consumption values of all agents are ambiguity--free in mean, i.e. their expectation does not depend on the specific prior in the representing set $\PP$. 

We then show for the particular transparent example of no aggregate uncertainty and Gilboa--Schmeidler preferences that this property is generically not satisfied. 

\begin{defi}
Fix a convex, compact, nonempty  set of priors $\mathbb P$. We call a plan 
$\xi \in \mathbb X$ $(\mathbb P)$--ambiguity free in mean if $\xi$ has the same expectation for all $Q \in\mathbb P$, i.e. there is a constant $c \in \mathbb R$ with
$E^Q \xi = c$ for all $Q \in \mathbb P$.
\end{defi}
Note that a  plan $\xi$ is ambiguity--free in mean if and only if we have
$$\EE (-\xi)= - \EE \xi\,.$$ We will use this fact sometimes below\footnote{
The concept has appeared before in (\cite{riedel2014non} and \cite{de2011ambiguity}. For unambiguous events, see also \cite{epstein2001subjective}.
In the spirit of \cite{brunnermeier2014welfare}, elements in the space $\mathbb{L}$, can be considered as belief neutral in expectation. Section 4 presents a more detailed account. A stronger notion would require that the probability distribution of a plan is the same under all priors in $\PP$; \cite{ghirardato2004differentiating} call such plans ''crisp acts''.}.

\begin{lem}\label{LemmaL}
The set of plans $\xi \in \mathbb X$ which are $\mathbb P$--ambiguity--free in mean forms a subspace of $\mathbb X$. We denote this subspace by $\mathbb L$ or $\mathbb L^{\mathbb P}$.
If $\# \mathbb P >1$, $\mathbb L$ has a strictly smaller dimension than $\mathbb X$.
\end{lem}

\begin{proof3}{\ref{LemmaL}}
Let $\mathbb L$ denote the set of all contingent plans which are ambiguity--free in mean. 
Constant plans are obviously ambiguity--free in mean, hence $\mathbb L$ is not empty. As expectations are linear, the property of being ambiguity--free in mean is preserved by taking sums and scalar products. Hence, $\mathbb L$ is a subspace of $\mathbb X$.

 If $\# \mathbb P >1$, we have $P_1, P_2\in\mathbb{P}$ such that $P_1-P_2 \neq 0\in \mathbb{X}$.  
In abuse of notation $x\in \mathbb{X}$ is $\{P_1,P_2\}$--ambiguity--free in the mean, if 
$$
\langle P_1, x \rangle=  \langle P_2, x \rangle.
$$
This equation yields a hyperplane $H=\{x: \langle P_1-P_2, x \rangle=0\}$, with $0\in H$. 
Consequently $H$ is subvector space of $\mathbb{X}$ with strictly smaller dimension and contains all  plans being $\{P_1,P_2\}$--ambiguity free in mean. 

The result follow from the first part and $\{P_1,P_2\}\subset \mathbb{P}$ implies $\mathbb{L}\subset H$. 
\end{proof3}

We can now clarify when Arrow--Debreu equilibria of a particular linear economy $\cE^{\{P\}}$ are also Knight--Walras equilibria.

\begin{thm}\label{ThmAD=KW}
Fix a prior $P \in \mathbb{P}$. Let $(\psi, \left(c^i\right))$ be an Arrow--Debreu equilibrium for the (linear) economy $\cE^{\{P\}}$. 
Then $(\psi,\left(c^i\right))$ is a Knight--Walras equilibrium for $\cE^\mathbb{P}$ if and only if  the value of net demands $\xi^i = \psi (c^i-e^i)$ are ambiguity--free in the mean for all agents $i$.
\end{thm}

\begin{proof}
Let $(\psi, \left(c^i\right))$ be an Arrow--Debreu equilibrium for the (linear) economy $\cE^{\{P\}}$.  
Then markets clear.

Suppose first that the value of net demands $\xi^i = \psi (c^i-e^i)$ are ambiguity--free in the mean for all agents $i$. 
We need to check that $c^i$ is in agent $i$'s budget set for the Knightian economy $\cE^{\mathbb P}$, and optimal.
By assumption, we have 
$$E^Q \psi (c^i-e^i) = c $$ for all $Q \in \mathbb{P}$ for some constant $c$.
As $c^i$ is budget--feasible in $\cE^\PP$ and utility functions are locally non--satiated by Assumption \ref{A2}, we have $c=0$, i.e.
$$ \mathbb E \psi (c^i-e^i) = E^P \psi (c^i-e^i) = 0 \,.$$
As $c^i$ is part of an Arrow--Debreu equilibrium, $c^i$ is optimal in the linear budget set given by the prior $P$; this budget set contains the budget of the Knightian economy $\cE^{\mathbb P}$. Hence, $c^i$ is optimal for agent $i$ in the Knightian economy. We conclude that $(\psi,\left(c^i\right))$ is a Knight--Walras equilibrium for $\cE^{\mathbb P}$.

Now suppose that $(\psi,\left(c^i\right))$ is a Knight--Walras equilibrium. 
 We need to check that all $\xi^i$ have expectation zero under all $Q \in \mathbb P$ for all $i$.

As utility functions are locally non satiated, the budget constraint is binding for all agents, $\mathbb E \mathbb \xi^i = 0$ for all $i$.
It is enough to show that $\mathbb E \mathbb (-\xi^i)=0$ for all $i$ (because this entails
$ \min_{P\in\mathbb P} E^P \xi^i = \max_{P\in\mathbb P} E^P \xi^i=0\,.$)
By sublinearity, we have $\mathbb E \mathbb (-\xi^i)\ge 0$.
Market clearing implies
$$ \mathbb E (- \xi^i) = \mathbb E \left( \sum_{j\not= i} \xi^j \right)
\le \sum_{j\not= i} \mathbb E \xi^j = 0\,.$$
We conclude that $\mathbb E (- \xi^i)=0$ for all $i$, as desired.  
\end{proof}

We now turn to the particularly transparent case of no aggregate uncertainty and multiple prior utilities. We shall show that generically in endowments, Arrow--Debreu equilibria are not Knight--Walras equilibria.

So assume for the rest of this section that $$ \sum e^i = 1$$ and utilities are of the form
$$U^i(c)= \min_{P \in \mathbb P} E^P u^i(c^i)$$ where the Bernoulli utility functions $u^i : \mathbb R_+ \to \mathbb R$ satisfy  the standard assumptions of Example \ref{ExUtility}.

These economies have been studied in detail in  (\cite{billot2000sharing} and \cite{chateauneuf2000optimal}). We recall the results.
if  $(\psi, (c^i))$ is an Arrow--Debreu equilibrium of the economy  $\left( I, \left(e^i,U^i \right)_{i\in I}, \mathbb{P} \right)$,
 $(c^i)$ is a full insurance allocation. Moreover, there exists $\tilde P \in\PP$ such that $\psi$ is proportional to the density $\frac{d\tilde P}{dP}$, and 
$(1, (c^i))$ is an equilibrium of the economy  $\left( I, \left(e^i,U^i \right)_{i\in I}, \{\tilde P\} \right)$.

Equilibrium prices are not determinate because the utility gradient at equilibrium consists of all priors $\mathbb P$. 
The above result states that one can take the equilibrium state-price density to be equal to $1$ after a suitable change of measure.

\begin{thm}\label{thm2}Assume that $\EE$ is not linear. 
Generically in endowments, Arrow--Debreu equilibria of $\cE^{\{P\}}$ for some $P \in \mathbb P$ are not  Knight--Walras equilibria of $\cE^{\mathbb P}.$ 

More precisely: let $M=\left\{ (e^i)_{i=1,\ldots,I} \in \mathbb X_{++}^I : \sum e^i = 1 \right\}$ be the set of economies without aggregate uncertainty normalized to $1$. 
Let $ N$ be the subset of elements $(e^i)$ of $M$ for which there exists $P \in \mathbb P$ and an Arrow--Debreu equilibrium $(\psi, (c^i))$ of 
 the economy  $\cE^{\{P\}}$ 
which is also a Knight--Walras equilibrium of the economy $ \cE^\EE$. $N$ is a Lebesgue null subset of the $(I-1) \cdot  \# \Omega$--dimensional manifold $M$.\end{thm}

\begin{proof}
Let $(e^i)$ be an allocation in $N$. Let $(\psi, (c^i))$ be an Arrow--Debreu equilibrium of the economy  $\left( I, \left(e^i,U^i \right)_{i\in I}, \{P\} \right)$. Then we know that $c^i >0$ is a constant for each agent $i$. Hence, each $P\in \mathbb{P}$ is a minimizing prior of maxmin expected utilities. In particular and in abuse of notation, we have $\partial U^i(c^i)=   {u^i}'(c^i)\cdot\mathbb{P}$, for each $i\in \mathbb{I}$.

By the first order condition of Pareto optimality of the $P$-Arrow-Debreu equilibrium, we have for some equilibrium weight $\alpha_P\in \Delta_I$
$$
\psi\cdot  P\in  \bigcap_{i\in \mathbb{I}} \alpha^i_P \partial U^i(c^i)=   \bigcap_{i\in \mathbb{I}} \alpha^i_P {u^i}'(c^i)\mathbb{P}.
$$
Note, that the intersection is nonempty by the Pareto optimality of $(c^i)$. 

Since $\alpha_i {u^i}'(c^i)$ is constant for $i\in I$ and $P\in\mathbb{P}$, this implies that $\psi$ is contant.
Let us say $\psi=1$ without loss of generality.
From Theorem \ref{ThmAD=KW}, we then know that $\psi (c^i-e^i) \in \mathbb L$. As $\psi$ and $c^i$ are constants, and $\mathbb L$ is a vector space by Lemma \ref{LemmaL}, we conclude that $e^i \in \mathbb L$. 
As the space $\mathbb L$ has strictly smaller dimension than $\mathbb X$, again by Lemma \ref{LemmaL}, we conclude that $N$ is a null set in $M$. 
\end{proof}

\section{Uncertainty--Neutral Efficiency}

In general, Knight--Walras equilibria are inefficient. We introduce now a concept of constrained efficiency for our Knightian framework. If the Walrasian  auctioneer aims for  \textit{robust} rules,  he might consider only net trades that are independent of the specific priors in $\PP$. 

We might also consider a situation of cooperative negotiation among the agents. In a framework of Knightian uncertainty described by the set of priors $\PP$, different priors may matter for different agents. For multiple prior agents, e.g., different priors are usually relevant for buyers and sellers of a contingent claim\footnote{For general variational preferences, one can still derive the relevant beliefs for a given contingent consumption  plan, compare, e.g., \cite{rigotti2008subjective}. 
 This welfare theorem can be compared with the welfare criterion  of \cite{blume2015case} for evaluation of different  financial market designs, which has a similar motivation as for  our notion of Knight--Walras equilibria.
On the one hand, we aim to prevent welfare-reducing speculation through  the disambiguation of net trades.    On the other hand, the price we have to pay, as a designer, is the exclusion of welfare-improving insurance possibilities through the conditions on the value of net trades. 
}.
A reallocation of goods over states is then uncontested  if  its value is independent\footnote{Recently, new notions of Pareto optimality under uncertainty appeared in the literature.  \cite{gilboa2014no}, \cite{brunnermeier2014welfare}, and \cite{blume2015case} make the point that the standard notion of  Pareto optimality can be spurious
when agents hold subjective heterogeneous beliefs. Earlier, \cite{dreze1972market} already emphasized that importing   the concept of Pareto dominance under certainty to the Arrow--Debreu world under uncertainty may result in odd implications.  Since our setup allows for heterogeneous expectations, see Example \ref{ExUtility}, their arguments carry over, in principle, to our setup.}  of the specific priors in $\PP$.

The preceding reasoning suggests the following concept of constrained efficiency.

\begin{defi}\label{PO}
Let $\cE=\left( I, \left(e^i,U^i \right)_{i\in \mathbb{I}}, \mathbb E \right)$ be a Knightian economy. 
Let $c=\left(c^i\right)_{i\in \mathbb{I}}$ be a feasible allocation. Let $\psi$ be a state--price density. We call the allocation $c$ \textnormal{uncertainty neutral   efficient} (given $\psi$ and $\EE$) if there is no other feasible allocation $d=\left(d^i\right)_{i=1,\ldots,I}$ with
$$ \eta^i= \psi \left(d^i-e^i\right) \in \mathbb L^\EE$$ and  
 $U^i(d^i)> U^i(c^i)$ for all $i\in \mathbb{I}$. 
\end{defi}

Knight--Walras equilibria satisfy our robust version of efficiency.

\begin{thm}\label{eqinef}
Let  $(\psi,c)$  be a Knight--Walras equilibrium of the Knightian economy $\cE=\left( I, \left(e^i,U^i \right)_{i\in \mathbb{I}}, \mathbb E \right)$. Then $c$ is uncertainty neutral  efficient (given $\psi$ and $\EE$).
\end{thm}

\begin{proof}
Let $(\psi,c)$ be a Knight--Walras equilibrium of the Knightian economy $\cE=\left( I, \left(e^i,U^i \right)_{i\in \mathbb{I}}, \mathbb P \right)$.  Suppose there is a feasible allocation $d=\left(d^i\right)_{i=1,\ldots,I}$ 
with $U^i(d^i)> U^i(c^i)$ for all $i\in \mathbb{I}$. From optimality, we have then $ d^i \notin \mathbb B(\psi,e^i)$, or $\EE \eta^i >0$.
Suppose furthermore 
$ \eta^i= \psi \left(d^i-e^i\right) \in \mathbb L^\EE$.
Take any prior $P \in \PP$. As the net excess demand is ambiguity--free in mean, we have
$$E^P \eta^i = \mathbb E \eta^i > 0\,.$$
As the expectation under $P$ is linear, we obtain by summing up and feasibility of the allocation $d$
$$ 0 = E^P \sum_{i=1}^I \psi (d^i-e^i) =  \sum_{i=1}^I  E^P \psi (d^i-e^i) >0 \,,$$ a contradiction. 
\end{proof}


\section{Sensitivity  of Arrow--Debreu Equilibria with respect to Knightian Uncertainty }

In this section we explore first the robustness of Arrow--Debreu equilibria with respect to the introduction of a small amount of Knightian uncertainty when agents have multiple prior utilities. With no aggregate uncertainty, equilibria change in a discontinuous way with small uncertainty perturbations; whereas agents attain full insurance under pure risk, no trade (and thus no insurance) occurs in equilibrium with a tiny amount of Knightian uncertainty. 
We then take the opposite view and consider growing uncertainty. When uncertainty is sufficiently large, no trade is again  the unique equilibrium.

Throughout this section, we fix continuously differentiable, strictly concave, and strictly increasing Bernoulli utility functions $u^i : \mathbb R_+ \to \mathbb R$ and write for a given set of priors $\PP$
$$U_\PP^i(c)= \min_{P\in\mathbb P} E^P u^i(c)$$ for the associated multiple prior utility function.\footnote{Insurance properties of  sharing rules  are considered for several classes of  preferences.  \cite{de2011ambiguity}  discuss the case for Choquet expected utility. A more general perspective can be found in     \cite{rigotti2008subjective}.}

Let us start  with an example where   the introduction of a tiny amount of uncertainty changes the equilibrium allocation drastically. 

\begin{ex}\label{ex4}
 Let  $\Omega=\{1,2\}$.  
Let the set of  priors be $\mathbb P_\epsilon=\{ p\in \Delta: p_1\in [\frac{1}{2}-\epsilon,\frac{1}{2}+\epsilon]\}$ for some  $\epsilon\in[0,1/2)$.

For $\epsilon>0$,  a consumption plan is ambiguity--free in mean if and only if it is  full insurance;  we have $\mathbb L^{\mathbb P}=\{c\in\mathbb{X}: c_1=c_2\}$. 

 Let there be no aggregate ambiguity, without loss of generality  $e=1$ in both states.
Let there be two agents $I=2$ (with multiple prior utilities as stated above) and uncertain endowments, e.g. $e^1=(1/3,2/3)$ and $e^2=(2/3,1/3)$.

In a Knight--Walras equilibrium, the state price has to be strictly positive because of strictly monotone utility functions. 
Since we have two agents, the budget constraint implies that 
$$0= \EE \psi (c^1-e^1)= \EE \psi (c^2-e^2)$$
or
$$ 0=\EE \psi (c^1-e^1) = \EE (-\psi (c^1-e^1))\,.$$
Hence, $\psi (c^1-e^1)$ is mean--ambiguity free, thus constantly equal to zero here. Since $\psi$ is strictly positive, $c^1=e^1$ follows. 
There is no trade in Knight--Walras equilibrium for every $\epsilon>0$. 
In sharp contrast, agents achieve full insurance in every Arrow--Debreu equilibrium of any linear economy $\cE^{\{P\}}$.
\end{ex}

The example uses the fact that we are in a simple world with two states and two agents. In general, the situation will be more involved. Nevertheless, the discontinuity when passing from a risk economy to $\cE^{\{P\}}$ to a Knightian economy $\cE^\EE$ remains.

Let us now consider economies of the form 
$$\mathcal{E}^{\mathbb{E}}=\left( I, \left(e^i,U^i_\P \right)_{i=1\ldots,I}, \mathbb E^\P \right)$$  with strictly positive initial endowment allocation $e=(e^1,\ldots, e^I)\in \mathbb{X}_{++}^I$. Here, $\EE^\PP$ denotes the Knightian expectation induced by the set of priors $\PP$. We assume that aggregate endowment is constant, $\bar{e}= \sum_{i\in I} e^i \in \mathbb R_{++}$.

Let  $\mathbb{K}(\Delta)$ be the set of closed and convex subsets of $\operatorname{int}(\Delta)$ equipped with   the usual  Haussdorff metric $\mathtt{d}_H$.
 Define the \textit{Knight--Walras equilibrium correspondence} $\mathcal{KW}: \mathbb{K}(\Delta)\times  \mathbb{X}_+^I\Rightarrow \mathbb{X}_+^{I+1}$ via 
$$\mathcal{KW}(\mathbb{P})=\Big\{ (\psi,c)\in \mathbb{X}_+^{I+1}: (c,\psi) \textnormal{ is a KW--equilibrium in }\mathcal{E}^{\mathbb{P}} \Big\}.$$
According to Theorem \ref{MAIN}, the set of  KW--equilibria $\mathcal{KW}(\mathbb{P})$ in  the economy  is nonempty.

\begin{theorem}
Let $\mathbb{P}: [0, 1)\rightarrow \mathbb{K}(\Delta)$ be a continuous  correspondence with $ \mathbb{P}(0)= \{P_0\} $ for some $P_0\in \operatorname{int}(\Delta)$. For $0<\epsilon<1$, assume $P_0 \in \operatorname{int} \mathbb{P}({\epsilon}) $. For $0\le \epsilon <1$, define the Knightian expectation  $\mathbb E_\epsilon X = \EE^{\PP(\epsilon)} X =  \max_{P \in \PP(\epsilon)} E^P X$.

The Knight--Walras equilibrium correspondence 
$$ \epsilon \mapsto \mathcal{KW}(\mathbb{P}(\epsilon), e) $$
is discontinuous in zero. 
\end{theorem}

\begin{proof}
For $\epsilon=0$, we are in an Arrow--Debreu economy without aggregate uncertainty. As a consequence, $\mathcal{KW}(\mathbb{P}(0), e)$ contains only   full insurance allocations.

Fix $\epsilon>0$. Let us first show that a  mapping $X : \Omega \to \mathbb R$ is $\PP(\epsilon)$--ambiguity--free  in mean if and only if it is constant. 
Due to our assumptions, $\PP(\epsilon)$ contains a ball around $P_0$ of the form 
$$B_\eta(P_0)=\left\{ Q \in \Delta : \| Q-P_0 \| < \eta \right\}$$ for some $\eta>0$. We use here, without loss of generality, the maximum norm in $\mathbb{R}^\Omega$.

Suppose $E^Q X = c$ for some $c \in \mathbb R$ and all $Q \in \PP(\epsilon).$  Let $1 = (1, 1, \ldots,1) \in \mathbb{X}$ denote the vector with all components equal to $1$.  Let $Z \in \mathbb{X}$ satisfy $Z \cdot 1 = 0$ with $\| Z \|=1$. Then $P_0 + \tilde{\eta} Z \in B_\eta(P_0) \subset \PP(\epsilon)$ for all $0<\tilde{\eta} < \eta$.  Hence, we have
$$c = E^{P_0 + \eta Z} X = E^{P_0} X + \tilde{\eta} Z \cdot X\,.$$ As $0<\tilde{\eta} < \eta$ is arbitrary, $Z \cdot X=0$ for all $Z$ with norm $1$ and $Z \cdot 1 =0$ follows. By linearity, this extends to  all $Z$ with $Z \cdot 1=0$; it follows that $X$ is a multiple of $1$, hence constant.

In the next step, we show that  $ (\psi,c)\in \mathcal{KW}(\mathbb{P}(\epsilon)) $ implies $c=e$. Let $(\psi,c)$ be a Knight--Walras equilibrium for the economy $\cE^{\P(\epsilon)}$. Let $\xi^i = \psi(c^i-e^i)$ be the value of net trade for agent $i$. Then $\sum \xi^i=0$ by market clearing. 
As the utility functions are strictly monotone, the budget constraint is binding, so $\EE_\epsilon \xi^i = 0$ for all $i$. From subadditivity, we get $\EE_\epsilon (- \xi^i ) \ge 0$. On the other hand, from market clearing, subadditivity, and the binding budget constraint, 
$$ \EE_\epsilon [-\xi^i]= \EE_\epsilon \sum_{j \not= i} \xi^j \le \sum_{j \not= i} \EE_\epsilon \xi^j = 0 \,.$$
We conclude that $\xi^i$ is  ambiguity--free in mean, thus constant. Due to the budget constraint, $\xi^i=0$. As state prices must be strictly positive in equilibrium due to strictly monotone  utility functions, we conclude that $c^i=e^i$.

The Knight--Walras equilibrium correspondence is thus discontinuous in zero.
\end{proof}

The previous result shows that a tiny amount of Knightian uncertainty can substantially change equilibria. We now consider the opposite case of growing Knightian uncertainty and impose no assumption on the aggregate endowment $\bar e=\sum e^i$. We show that no trade is the only equilibrium if Knightian uncertainty is large enough, thus generalizing our initial example \ref{ex2}.

Next we state a simple result on uniqueness of Knight--Walras equilibria, when no--trade is an equilibrium. 

\begin{lem}\label{notrade}
If $(\psi, e)$ is a Knight--Walras equilibrium, then $e$ is the unique Knight--Walras equilibrium allocation.
\end{lem}

\begin{proof}
Suppose there is  another Knight--Walras equilibrium allocation $(\psi',x)$ with $\emptyset\neq \mathbb{J}=\{i\in \mathbb{I}: x^i\neq e^i\}$. We have $U^j(x^j)\geq U^j(e^j)$ for all $j\in \mathbb{J}$.

We  show $ \mathbb{E}[\psi'(x^j-e^j)]>0$ for all $j\in \mathbb{J}$, which contradicts the budget feasibility of $x^j$. Take some $\epsilon>0$ and note that $U^j(x^j+\epsilon e^j)>U^j(x^j)$ by  strict monotonicity. As $x^j$ is optimal in the budget set, we obtain   $ \mathbb{E}[\psi'(x^j + \epsilon e^j -e^j)]>0$. Letting $\epsilon$ to zero, we have $ \mathbb{E}[\psi'(x^j- e^j)]\geq0$. 
Now suppose $ \mathbb{E}[\psi'(x^k- e^k)]=0$  for some $k\in \mathbb{J}$. Since $U^k$ is strictly concave, we derive for any $\alpha\in(0,1)$
$$
U^k(\alpha x^k+(1-\alpha)e^k)> \alpha U^k( x^k)  + (1-\alpha) U^k(e^k)\geq U^k(e^k).
$$
We now obtain 
\begin{eqnarray*}
0&<& \mathbb{E}\left[\psi'\left(\alpha x^k+ (1-\alpha)e^i- e^k\right)\right]\\
&=&  \mathbb{E}\left[\psi' \alpha ( x^k- e^k)\right]\\
&=& \alpha \mathbb{E}\left[\psi'  ( x^k- e^k)\right]\\
&\leq&0\,,
\end{eqnarray*}
a contradiction.
\end{proof}

Next we increase the Knightian uncertainty  in the economy  $\mathcal{E}^\mathbb{P}$. 
As the following result shows, if ambiguity  becomes  sufficiently large  then  there is no trade in equilibrium.  Recall that we keep the standing assumption on multiple prior utility functions. 

\begin{thm}
   If ambiguity is large, every Knight--Walras--equilibrium is a no--trade equilibrium:
There is a $\mathbb{P}'\in \mathbb{K}(\Delta)$ such that for every  $\mathbb{P}''\in  \mathbb{K}(\Delta)$ with  $\mathbb{P}''\supset \mathbb{P}' $ we have  
$$\mathcal{KW}(\mathbb{P}'') = \mathbb X_{\psi,P''} \times \{e\}$$
for
$\mathbb X_{\psi,P''} = \left\{ \psi \in \mathbb X_{++}: {u^i}'(e^i)\cdot \argmax{P\in \mathbb{P}}E^P[u^i(e^i)]\cap \psi\cdot \mathbb{P}''  \:\:\forall i\in\mathbb{I}\right\} 
$.
\end{thm}

\begin{proof}Since utility is strictly increasing,  an equilibrium state price must be strictly positive. 

Under full Knightian uncertainty, $\PP=\Delta$, the budget set is
 $[0,e^i]$.
By  strict monotonicity and convexity of preferences,  the better--off set $\{x\in \mathbb{X}_+: U^i_\P(x)\geq U^i_\P(e^i)\}$ can be supported by a hyperplane with a strictly positive normal vector $\pi^i$.  Since $U^i_\P$ is of multiple prior type,  an increase of $\mathbb{P}$ to $\mathbb{P}'\in\mathbb{K}(\Delta)$  let the better--off set $\{x\in \mathbb{X}_+: U^i_{\P'}(x)\geq U^i_{\P'}(e^i)\}$ shrink and  $P^{\pi^i}= \frac{\pi^i}{\Vert \pi^i\Vert}$ remains a supporting prior.

For large $\mathbb{P}'$ such that  $P^{\pi^i}\in\mathbb{P}'$ for all $i\in \mathbb{I}$,    all individual  first order conditions are satisfied.  $e^i$ is then optimal in $\mathbb{B}^{\mathbb{P}' }(1,e^i)$. A larger $\mathbb{P}''\supset  \mathbb{P}'$ leave this result unchanged.  An application of Lemma \ref{notrade} establishes  uniqueness of  the no--trade equilibrium.
\end{proof}

 \section{Conclusion}
 Knightian uncertainty leads naturally to nonlinear expectations derived from a set of priors.  This led us to study a new equilibrium concept, Knight--Walras equilibrium, where prices are sublinear. 
We  established existence of such equilibrium points and studied its efficiency properties. While one cannot expect fully efficient allocations, in general, Knight--Walras equilibrium allocations satisfy a restricted efficiency ctriterion: when the social planner is restricted to ambiguity--neutral trades, she cannot improve upon a Knight--Walras equilibrium allocation.

The introduction of Knightian friction on the price side rather than the utility side can have strong effects. In a world without aggregate uncertainty, no trade equilibria result even with a tiny amount of uncertainty. 
The abrupt change of equilibria with respect to Knightian uncertainty has potentially strong implications for consumption--based asset pricing results which rely on the assumption of probabilistically sophisticated agents and markets. In dynamic and continuous--time models, these questions remain to be explored.

\begin{appendix}

\section{Existence}

We begin with an investigation of the Knight-Walras correspondence $\mathbb{B}$ in (\ref{budget}).  To prove the continuity of our budget correspondence, we follow the lines of \cite{debreu1993existence}. Set $[0,x]=\{y\in\mathbb{X}: 0\leq y \leq x \}\subset \mathbb{X} $. $\bar  e=\sum_i e^i$ denotes the aggregate endowment and $\Delta=\Delta_\mathbb{X}$ is the simplex in $\mathbb{X}$.

\begin{lem}\label{budprop}
\begin{enumerate}
\item The budget sets $\mathbb{B}(\psi,e^i)$ in (\ref{budget}) with $\psi \geq 0 $ are nonempty, closed, and convex. If $\psi$ and $e^i$ are strictly positive, $\mathbb{B}(\psi,e^i) $  is also compact.
\item The Knight-Walras budget correspondence  $\mathbb{B}(\cdot, e^i)$ is homogeneous of degree zero.
\item Let  the consumption set be $[0,2\bar e]$ and   fix any $e^i\in\mathbb{X}_{++}$.  Then the correspondence $\psi \mapsto \bar{\mathbb{B}}(\psi, e^i)= \mathbb{B}(\psi, e^i)\cap [0,2\bar e]$ is  continuous at any $\psi\in\Delta$. 
\end{enumerate}
\end{lem}

\begin{proof3}{\ref{budprop}}
\begin{enumerate}
\item Since $0,e^i\in \mathbb{B}(\psi,e^i)$, for every $\psi,e^i\in \mathbb{X}_+$, $\mathbb{B}$ is nonempty.   The budget set  $\mathbb{B}(\psi,e^i)$  is the intersection of  budget sets under linear prices of the form  $E^P[\psi \cdot]$, that is, 
$$
 \mathbb{B}(\psi,e^i)=\bigcap_{P\in\mathbb{P}} \mathbb{B}^P(\psi,e^i),
$$
where $ \mathbb{B}^P(\psi,e^i)= \left\{c \in \mathbb X_+ :  E^P[ \psi (c-e^i)] \le 0 \right\}$ denotes the closed and convex budget in an Arrow--Debreu economy under $\mathbb{P}=\{P\}$. The arbitrary intersection of convex (closed) sets is again convex (closed) and so is $ \mathbb{B}(\psi,e^i)$.

Standard arguments with linear price systems yield compactness of  $ \mathbb{B}^P(\psi,e^i)$ whenever $e^i(\omega),\psi(\omega)>0$ for all $\omega\in\Omega$.
 Since the arbitrary intersection of compact sets is again compact, the result then follows. 
\item  By definition, the nonlinear expectation is  positive homogeneous. The result then follows by the same arguments as in the case with linear price systems.
\item $[0,2\bar e]$ is a compact, convex,  nonempty   consumption set  in $\mathbb{X}=\mathbb{R}^\Omega$. 
We prove the continuity of $\bar {\mathbb{B}}:\Delta\Rightarrow [0,2\bar e]$.
 
To establish upper hemi-continuity, it suffices to show the closed graph property, since $\bar{\mathbb{B}}(\psi,e^i)$ is compact valued.
The graph of the  budget correspondence $gr(\bar{\mathbb{B}})=\{(\psi,x)\in \Delta\times [0,2\bar e]:x\in\bar{\mathbb{B}}(\psi,x) \}$  is closed since $\psi \mapsto  \max_{P\in\mathbb{P}}E^P\psi x$   is    continuous for all $x\in\mathbb{X}$, by  an application of Berge's maximum theorem. 
\\[.5em]
We show  lower hemi-continuity, i.e., if $(\psi_n,x_n)\rightarrow (\psi,x)$ in $\mathbb{X}\times \mathbb{X}$ and $x\in \bar{\mathbb{B}}(\psi,e^i)$ then there is a sequence $x_n\in [0,2\bar e]$ converging to $x$ and $x_n\in\bar{\mathbb{B}}(\psi_n,e^i)$, for every  $n\in\mathbb{N}$.
\newline
 Let us denote by $ \Psi_n$ the price system induced by $\psi_n\in\Delta$. We consider two cases.

\textnormal{case 1:} If $ \Psi(x-e^i)<0$, then for a large $\bar n$ we still have  $\Psi_{\bar n} (x-e^i)< 0$. We may take the following converging sequence
  \[
     x_n=\left\{\begin{array}{ll}x' \in \bar{\mathbb{B}}(\psi_n,e^i) \:\text{ arbitrary}, &\text{ if } n\leq \bar n \\
         x, & \:\:\:\:\:\:n>\bar n\end{array}\right. 
  \]
\textnormal{case 2:}  If $ \Psi(x-e^i)=0$, there is a $x'\in[0,2\bar e]$ such that $\Psi(x'-e^i)< 0$. Since $e^i>0$ by assumption,  take $x'=\frac{e^i}{2}$ and  since $ \mathbb{E}$  is  positive homogeneous, strictly monotone 
and constant preserving, we get
\begin{eqnarray*}
\Psi(x'-e^i)= \frac{1}{2}\Psi(-e^i)=  \frac{1}{2}\mathbb{E}-\psi e^i< \mathbb{E} 0=0.
\end{eqnarray*}
For $n$ large, the intersection of $ \{y\in \mathbb{X}: \Psi_n(y-e^i)=0\}$ and  $\{y\in \mathbb{X}:\exists\lambda \in\mathbb{R}: y=\lambda x+(1-\lambda) x'\}$ is nonempty and denoted by $ \Psi_n^\cap$. 
 Since $ \Psi_n^\cap$ is  the closed   subset of a line,   $\bar x_n=\arg\min_{y\in \Psi_n^\cap} \Vert y-x\Vert $ is unique.
 Now, set
  \[
     x_n=\left\{\begin{array}{ll}\bar x_n, &\textit{ if } \bar x_n\in [x',x] \\
         x, & \textit{ else}\end{array}\right. 
  \]
By construction, we have $x_n \in {\mathbb{B}}(\psi_n,e^i)$ and $x_n\rightarrow x$ in $\mathbb{X}$.
\end{enumerate}
\end{proof3}

 Let us introduce the \textit{Knight--Walras demand correspondence}
\begin{eqnarray*}
X^i(\psi,e^i):=\textnormal{arg}\!\!\!\!\max_{x\in  {\mathbb{B}}(\psi,e^i)}U^i(x) 
\end{eqnarray*}
and the aggregate excess demand 
\begin{eqnarray}\label{Z}
z(\psi)=\sum_{i\in\mathbb{I}} X^i(\psi,e^i)-e^i=\left\{z=\sum_{i\in\mathbb{I}}  x^i -e^i: x^i\in  X^i(\psi,e^i), i\in\mathbb{I}  \right\}.
\end{eqnarray}
We collect, under our standing Assumption \ref{A2},  the following standard properties for the aggregate demand correspondence which we employ in the proof of Theorem \ref{MAIN}.
\begin{pro}\label{Demand} Suppose Assumption \ref{A2} holds for every agent in the economy and let $\Psi(\cdot)=\mathbb{E}[\psi\cdot]$ be a sublinear price system on $\mathbb{X}$. The correspondence $z: \mathbb{X}_+\Rightarrow  \mathbb{X}$ in (\ref{Z})
\\[-1.5em]
\begin{enumerate}
\item  is  upper hemi-continuous and nonempty compact convex valued.
\\[-1.5em]
\item  is  homogeneous of degree zero.
\\[-1.5em]
\item  satisfies the weak Walras law: $\Psi(z)\leq 0$, for every $z\in z(\psi)$.
\end{enumerate}
\end{pro}

\begin{proof4}{\ref{Demand}}
\begin{enumerate}
\item The budget correspondence is convex valued and continuous, by Lemma \ref{budprop}. Applying Berge's Maximum theorem to each $X^i(\cdot)$ gives that $z(\cdot)$ is an upper hemi-continuous correspondence as well. The concavity of  $U^i$ gives us that $z(\cdot)$ is convex valued.
\item By Lemma \ref{budprop}, $\mathbb{B}$ is homogeneous of degree zero in $\psi$. This implies  that  each $X^i$ is also  homogeneous of degree zero, and so is $z(\cdot)$.
\item  By the budget constraint we have  $\Psi\left(x^i-e^i\right)\leq0$,  for every $x^i\in X^i(e^i,\Psi)$ and $i\in \mathbb{I}$. The sub-additivity of $\Psi$ then implies
\begin{eqnarray*}
\Psi(z )= \Psi\left(\sum_{i\in \mathbb{I}}x^i-e^i\right)\leq \sum_{i\in \mathbb{I}} \Psi\left(x^i-e^i\right)\leq0.
\end{eqnarray*}
\end{enumerate}
\end{proof4}

\begin{proof2}{\ref{MAIN}}
By the definition of  $\mathbb E$, every price system $ \mathbb{E}[\psi \cdot]: \mathbb{X}\rightarrow \mathbb{R}$ can be written as 
\begin{eqnarray*}
 \mathbb{E}\psi X= \max_{P\in\mathbb{P}}E^P\psi X, \qquad X\in\mathbb{X}.
\end{eqnarray*}
The budget set $\mathbb{B}(\psi,e^i)$ is in general not
 compact within $\mathbb{X}$, so we truncate $\mathbb{B}$ via $\overline{\mathbb{B}}(\psi,e^i)
=\mathbb{B}(\psi,e^i)\cap [0,2\bar e]$ and denote the corresponding truncated economy by
 $\overline{ \mathcal{E}}=\{[0,2\bar e], \bar U^i,e^i\}$.  $\bar U^i:[0,2\bar e]\rightarrow \mathbb{R}$ is the restriction of $U^i$ to $[0,2\bar e]$. 
We show first existence of an equilibrium in $\overline{ \mathcal{E}}$ and then check in step 6 and 7 that this $I+1$ tuple is also an equilibrium in ${ \mathcal{E}}$.

The  existence  proof of an equilibrium  in $\overline{\mathcal{E}}$ is divided into  five  steps. 
\begin{enumerate}
\item{\textit{ \: Continuity of the  Budget correspondence:}} By Assumption \ref{A2}, each initial endowment $e^i$ is strictly positive.  The continuity of the correspondence $\overline{\mathbb{B}}:\Delta\Rightarrow [0,2\bar e]$ follows  from  Lemma \ref{budprop}.3.
\item 
{\textit{ \: Demand correspondence:}}
Consider the (truncated) demand correspondence $\overline {X}^i:\Delta\Rightarrow [0,2\bar e]$. By step 1, $\overline{\mathbb{B}}(\cdot, e^i):\Delta\Rightarrow [0,2\bar e]$ is continuous, hence by Berge's maximum theorem the demand 
\begin{eqnarray*}
\overline {X}^i(\psi)= \textnormal{arg}\!\!\!\max_{x\in\overline{\mathbb{B}}
(\psi,e^i)}  \bar U^i\left(x\right)
\end{eqnarray*}
is upper hemi-continuous, compact and non empty valued, since $\overline{U}^i$ is continuous on $gr(\overline{\mathbb{B}})$. By the   concavity of $U^i$, $\overline {X}^i(\psi)$  is convex-valued.

\item { \textit{\: Walrasian Player:}} 
 Define the Walrasian price adjustment correspondence  $W: [0,2\bar e]^I\times \mathbb{P}\Rightarrow   \Delta $  via
\begin{eqnarray*}
W (  {x^1},\ldots,   {x^I}, P)= \textnormal{arg}\max_{\psi\in \Delta }E^{P}\left[\psi  \sum_{i\in\mathbb{I}}  \Big ( x^i-e^i\Big)\right].
\end{eqnarray*}
Again, by Berge's Maximum Theorem, the correspondence is  upper hemi-continuous and attains convex, compact and nonempty values, by the continuity of the linear expectation  in $(x^1,\ldots,x^I)$ and the  linearity  in $\psi$. 
 
\item {   \textit{\: Knightian Price Players:}} Define the Knightian adjustment  correspondence   $K:[0,2\bar e]^I\times \Delta\Rightarrow \mathbb{P} $ via 
$$K\left(   {x^1},\ldots,   {x^I}, \psi\right)= \textnormal{arg} \max_{P\in\mathbb{P}} E^{P}\left[\psi  \sum_{i\in\mathbb{I}}  \Big ( x^i-e^i\Big)\right].$$
Once gain, by Berge's Maximum Theorem, the correspondence is  upper hemi-continuous and attains convex, compact and nonempty values.
\item 
{ \: \: \textit{\: Existence of a Fixed-Point:}} Set $\overline{X}= \left(\overline{X}^1,\ldots, \overline{X}^I\right)$. Putting things together we have the combined correspondence
\begin{eqnarray*}
\left[K\! \overline{X}\!W\right]:\mathbb{P}\times [0,2\bar e]^I\times \Delta\Rightarrow \mathbb{P}\times [0,2\bar e]^I\times \Delta
\end{eqnarray*}
as a product of nonempty and compact--convex  valued upper hemi-continuous correspondences (see step 2, 3 and 4).  Consequently,   a fixed--point
\begin{eqnarray*}
\left(\bar P ,\bar x^1,\ldots \bar x^i, \bar \psi\right ) \in \left[K\! \overline{X}\!W\right]\left(\bar P,\bar x^1,\ldots \bar x^i, \bar \psi\right)
\end{eqnarray*}
exists 
by  an application of Kakutani's fixed-point theorem.
\item { \:  \textit{\: Feasibility:}} We check the feasibility of the fixed-point allocation $\bar x$.

By the budget constraint and  the sublinearity of $X\mapsto \mathbb{E}\bar\psi X$ (since  $\bar\psi\geq 0$), we derive by the definitions of $\bar\psi$, via the Walrasian price player correspondence $W$, and $\bar P$, via the Knightian price player correspondence $K$,
\begin{eqnarray}\label{feas}
0&\geq &\sum_i  \mathbb{E}\Big[\bar \psi(\bar x^i-e^i)\Big]\nonumber\\
&\geq &    \mathbb{E}^{\mathbb{P}}\Big[\bar \psi \sum_i(\bar x^i-e^i)\Big] \nonumber\\
&=&  {E}^{\bar{P}}\Big[\bar \psi \sum_i(\bar x^i-e^i)\Big]  \nonumber\\
&\geq&  {E}^{\bar{P}}\Big[ \psi \sum_i(\bar x^i-e^i)\Big] 
\end{eqnarray}
The first inequality follows from the definition of the budget set and $\bar x^i \in \bar X^i (\psi)$ for all $i\in \mathbb{I}$.
The last inequality holds for all $\psi \in \Delta$ and by the positive homogeneity of linear expectations, it holds   even for all $\psi \in \mathbb{X}_+$. Since $\bar \in \operatorname{int}(\Delta)$.  We have $l\left(\sum_{i\in\mathbb{I}}\bar x^i-e^i)\right)\leq 0$ for all positive linear form on $\mathbb{X}$. This implies $ \sum_{i\in\mathbb{I}} (\bar x^i-e^i)\leq 0$. 


For the feasibility of the equilibrium allocation, the truncation is irrelevant.  
\item \textit{ Maximality in $\mathcal{E}^\mathbb{P}$:}
 Since $\bar x^i\in \overline {X}^i(\bar\psi)$, we have
\begin{eqnarray*}
\bar x^i\in \textnormal{arg}\!\!\!\!\!\!\!\!\max_{x\in\mathbb{B}(\psi,e^i)\cap [0,2\bar e]}\bar U^i(x).
\end{eqnarray*}
We have to show that $\bar x^i $ also maximizes $U^i$  on $\mathbb{B}(\psi,e^i)$. Suppose there  is  a $x\in 
\mathbb{B}(\psi,e^i)$ in the original  budget set, such that $ U^i(x)>U^i(\bar x^i)$. But then we have for some
 small $\lambda\in (0,1)$, $\lambda x +(1-\lambda) x^i \in  \mathbb{B}(\psi,e^i)\cap
 [0,2\bar e]$. The concavity of each $U^i$ gives us
\begin{eqnarray*}
U^i(\lambda x +(1-\lambda) x^i )\geq \lambda U^i(x) +(1-\lambda)U^i( x^i )>U^i( x^i ),
\end{eqnarray*}
a contradiction. Therefore, $(\bar x^1,\ldots,  \bar x^I,  \bar \psi)$ is  also an equilibrium in the original economy $\mathcal{E}^\mathbb{P}$. 
\end{enumerate}

\end{proof2}



\end{appendix}
 \bibliographystyle{econometrica}
\ifx\undefined\BySame
\newcommand{\BySame}{\leavevmode\rule[.5ex]{3em}{.5pt}\ }
\fi
\ifx\undefined\textsc
\newcommand{\textsc}[1]{{\sc #1}}
\newcommand{\emph}[1]{{\em #1\/}}
\let\tmpsmall\small
\renewcommand{\small}{\tmpsmall\sc}
\fi


\end{document}